\documentclass[11pt]{article}
\title{Lower bounds for 2-query LCCs over large alphabet}

\usepackage{xspace}
\usepackage{verbatim}
\usepackage{color}
\usepackage{graphicx}
\usepackage{tabularx}

\usepackage{amsmath}
\usepackage{bbm}
\usepackage[showonlyrefs]{mathtools}
\usepackage{mathstyle}
\usepackage{empheq}
\usepackage{amstext,amssymb,amsfonts}
\usepackage{fullpage}
\usepackage{nicefrac}
\usepackage{bm}

\usepackage{algorithmic,algorithm}

\usepackage{todonotes}

\newcounter{mynotes}
\usepackage{textcomp,setspace}

\usepackage{nameref}
\usepackage[pagebackref,colorlinks,linkcolor=blue,filecolor = blue,citecolor = blue, urlcolor = blue, hyperfootnotes=false]{hyperref}
\usepackage{color}
\usepackage[capitalize, noabbrev]{cleveref}

\usepackage{amsthm}
\usepackage{thmtools}

\renewcommand{\qedsymbol}{$\square$}

\declaretheorem[within=section]{theorem}

\declaretheorem[sibling=theorem]{lemma}
\declaretheorem[sibling=theorem]{claim}

\declaretheorem[sibling=theorem]{definition}

\newcounter{termcounter}
\renewcommand{\thetermcounter}{\Alph{termcounter}}
\crefname{term}{term}{terms}
\creflabelformat{term}{\textup{(#2#1#3)}}

\makeatletter
\def\term{\@ifnextchar[\term@optarg\term@noarg}
\def\term@optarg[#1]#2{%
  \textup{(#1)}%
  \def\@currentlabel{#1}%
  \def\cref@currentlabel{[][2147483647][]#1}%
  \cref@label[term]{#2}}
\def\term@noarg#1{%
  \refstepcounter{termcounter}%
  \textup{(\thetermcounter)}%
  \cref@label[term]{#1}}
\makeatother

\newcommand{\ignore}[1]{}

\newcommand{\bits}{\{0,1\}}

\newcommand{\BC}{\bits}

\newcommand{\supp}{\mathrm{supp}}

\newcommand{\poly}{\mathrm{poly}}

\newcommand{\vc}{\mathrm{vc}}

\newcommand{\norm}[1]{\lVert#1\rVert}

\newcommand{\outdeg}{\mathrm{deg}^+}
\newcommand{\indeg}{\mathrm{deg}^-}

\newcommand{\vecarrow}[1]{\overrightarrow{#1}}

\definecolor{DSred}{rgb}{1,0,0}

\renewcommand{\leq}{\leqslant}
\renewcommand{\geq}{\geqslant}
\renewcommand{\ge}{\geqslant}
\renewcommand{\le}{\leqslant}
\renewcommand{\epsilon}{\varepsilon}
\newcommand{\eps}{\epsilon}

\newcommand{\F}{\mathbb{F}}

\newcommand{\cA}{\mathcal A}

\newcommand{\cC}{\mathcal C}
\newcommand{\cD}{\mathcal D}

\newcommand{\cH}{\mathcal H}

\newcommand{\cM}{\mathcal M}

\newcommand{\cP}{\mathcal P}
\newcommand{\cQ}{\mathcal Q}
\newcommand{\cR}{\mathcal R}

\newcommand{\cX}{\mathcal X}

\newcommand{\barA}{\bar{A}}

\newcommand{\Esymb}{{\bf E}}

\newcommand{\Psymb}{{\bf Pr}}

\DeclareMathOperator*{\E}{\Esymb}

\DeclareMathOperator*{\ProbOp}{\Psymb}

\renewcommand{\Pr}{\ProbOp}

\def\notes{0}
 \newcommand{\gnote}[1]{\ifnum\notes=1{{\sf\color{red} [Gopi comment: #1]}}\fi}
 \newcommand{\anote}[1]{\ifnum\notes=1{{\sf\color{blue} [Arnab comment: #1]}}\fi}
 \newcommand{\avnote}[1]{\ifnum\notes=1{{\sf\color{blue} [Avishay comment: #1]}}\fi}

\usepackage[margin=1.in]{geometry}

\date{}
\author{
Arnab Bhattacharyya\thanks{Research partially supported by a DST Ramanujan Fellowship.}\\
Indian Institute of Science\\
{\small \texttt{arnabb@csa.iisc.ernet.in}}
\and 
Sivakanth Gopi\thanks{Research partially supported by NSF grants CCF-1523816 and CCF-1217416. Part of this research was done while the author was at Microsoft Research, Redmond.}\\
Princeton University\\
{\small \texttt{sgopi@cs.princeton.edu}}
\and 
Avishay Tal\thanks{Research supported by the Simons Collaboration on Algorithms and Geometry, and by the National Science Foundation grant No. CCF-1412958.}\\
Institute for Advanced Study\\
{\small \texttt{avishay.tal@gmail.com}}
}

\begin{document}
\maketitle
\thispagestyle{empty}
\begin{abstract}
A locally correctable code (LCC) is an error correcting code that allows correction of any arbitrary coordinate of a corrupted codeword by querying only a few coordinates.
We show that any $2$-query locally correctable code $\mathcal{C}: \{0,1\}^k \to \Sigma^n$ that can correct a constant fraction of corrupted symbols must have $n \geq \exp(k/\log|\Sigma|)$ under the assumption that the LCC is {\em zero-error}. We say that an LCC is zero-error if there exists a non-adaptive corrector algorithm that succeeds with probability $1$ when the input is an uncorrupted codeword. All known constructions of LCCs are zero-error. 

Our result is tight upto constant factors in the exponent. The only previous lower bound on the length of 2-query LCCs over large alphabet was $\Omega((k/\log|\Sigma|)^2)$ due to Katz and Trevisan (STOC 2000). Our bound implies that zero-error LCCs cannot yield $2$-server private information retrieval (PIR) schemes with sub-polynomial communication. Since there exists a $2$-server PIR scheme with sub-polynomial communication (STOC 2015) based on a zero-error $2$-query locally decodable code (LDC), we also obtain a separation between LDCs and LCCs over large alphabet. 
\ignore{
For our proof of the result, we need a new decomposition lemma for directed graphs that may be of independent interest. Given a directed graph $G$, our decomposition uses the directed version of the Szemer\'edi regularity lemma due to Alon and Shapira (STOC 2003) to find an equipartition of almost all of $G$ into subgraphs which are either edge-expanding or empty.}
\end{abstract}
\newpage
\setcounter{page}{1}
\section{Introduction}\label{sec:intro}
In this work, we study error-correcting codes that are equipped with local algorithms. A code is called a \textsf{locally correctable code (LCC)} if there is a randomized algorithm which, given an index $i$ and a received word $w$ close to a codeword $c$ in Hamming distance, outputs $c_i$ by querying only a few positions of $w$. The maximum number of positions of $w$ queried by the local correction algorithm is called the \textsf{query complexity} of the LCC. 

The main problem studied regarding LCCs is the tradeoff between their query complexity and length. Intuitively, these two parameters enforce contrasting properties. Small query complexity means that individual codeword symbols carry substantial information, while short length along with resilience to corruption means that information is spread out among the codeword symbols. In this paper, we explore one end of the spectrum of tradeoffs by studying $2$-query locally correctable codes.

Also called ``self-correction'', the idea of local correction originated in works by Lipton \cite{Lip90} and by Blum and Kannan \cite{BK89}  on program checkers. In particular, \cite{Lip90, BF90} used the fact that the Reed-Muller code is locally correctable to show average-case hardness of the Permanent problem. LCCs are closely related to \textsf{locally decodable codes (LDCs)}, where the goal is to recover a symbol of the underlying message when given a corrupted codeword using a small number of queries \cite{KT00}. LDCs are weaker than LCCs, in the sense that any LCC can be converted into an LDC while preserving relevant parameters (see Appendix~\ref{app:LCC vs LDC} for a formal statement and proof). LDCs and LCCs have found applications in derandomization and hardness results \cite{STV01, DS07,KS09}.  See \cite{Yek11} for a detailed survey on LDCs and LCCs, as of 2010. In more recent years, the analysis of LDCs and LCCs has led to a greater understanding of basic problems in incidence geometry, the construction of design matrices and the theory of matrix scaling, e.g. \cite{BDWY11, DSW12, DSW14}.

One particularly important feature of LDCs is their tight connection to {\em information-theoretic private information retrieval (PIR)} schemes. PIR is motivated by the scenario where a user wants to retrieve an item from a database without revealing to the database owner what item he is asking for. Formally, the user wants to retrieve $x_i$ from a $k$-bit database ${\bf x} = (x_1, \dots, x_k)$. A trivial solution is for the database owner to transmit the entire database no matter what query the user has in mind, but this has a huge communication overhead. Chor et al.~\cite{CGKM98} observed that while with one database, nothing better than the trivial solution is possible, there are non-trivial PIR schemes if multiple servers can hold replicas of the database. It turns out that $t$-server PIR schemes with low communication are roughly equivalent to short $t$-query LDCs.
 More precisely, a $2$-server PIR scheme for $k$ bits of data with $s$ bits of communication translates to a $2$-query LDC $\cC: \BC^k \to \Sigma^{2^s}$ where $\Sigma = \BC^s$. Note that in this translation, $|\Sigma|$ equals the length of the code.

Let $\mathcal{C}:\BC^k \to \Sigma^n$ be a 2-query LDC/LCC such that the corrector algorithm can tolerate corruptions at $\delta n$ positions. Katz and Trevisan in their seminal work \cite{KT00} showed that for $2$-query LDCs, $n \geq \Omega(\delta (k/\log |\Sigma|)^2)$. (Since LDCs are weaker than LCCs, a lower bound on the length of LDCs also implies a lower bound on the length of LCCs).  More than 15 years later,                                                                                                                                                                                                                                                                                                                                                                                                                                                                                                                                                                                                                                                                                                                                                                                                                                                                                                                                                                                                                                                                                                                                                                                                                                                                                                                                                                                                                                                                                                                                                                                                                                                                                                                                                                                                                                                                                                                                                                                                                                                                                                                                                                                                                                                                                                                                                                                                                                                                                                                  the Katz-Trevisan bound is still the best known for large alphabet $\Sigma$. However for small alphabet size, the dependence on $k$ is shown to be exponential. Goldreich et al. \cite{GKST02} showed that $n \geq \exp(\delta k/|\Sigma|)$ for linear 2-query LDCs, while Kerenedis and de Wolf \cite{KdW03} (with further improvements in \cite{WdW05}) showed using quantum techniques that $n \geq \exp(\delta k/|\Sigma|^2)$ for arbitrary 2-query LDCs. But these lower bounds become trivial when $|\Sigma|=\Omega(n)$. However, the case of large alphabet $|\Sigma| \approx n$ is quite important to understand as this is the regime through which we would be able to prove lower bounds on the communication complexity of PIR schemes. 

Given the lack of progress on LDC and PIR lower bounds, it is a natural question to ask whether strong lower bounds are possible for LCCs. In this work, we demonstrate an exponential improvement on the Katz-Trevisan bound for {\em zero-error LCCs}. We define a zero-error LCC to be an LCC for which the corrector algorithm is non-adaptive and succeeds with probability 1  when the input is an uncorrupted codeword. All current LCC constructions are zero-error, and in fact, any linear LCC can be made zero-error. 

\begin{theorem}[Informal]\label{thm:main}
If $\cC: \BC^k \to \Sigma^n$ is a zero-error $2$-query LCC that can correct $\delta n$ corruptions, then $n \geq \exp(\poly(\delta)\cdot k/\log|\Sigma|)$.\footnote{An earlier version \cite{BGprelim} of this paper showed that $n \geq \exp(c_\delta\cdot k/\log|\Sigma|)$ where $c_\delta$ has tower type dependence on $\delta$ due to the use of  the Szemer\'edi regularity lemma.}
\end{theorem}


\subsection{Discussion of Main Result}
The lower bound in \cref{thm:main} is tight in its dependence on $k$ and $\Sigma$. Specifically, Yekhanin in the appendix of \cite{BDSS16} gives the following elegant construction of a $2$-query LCC $\cC: \BC^k \to \Sigma^n$ with $n = 2^{O(k/\log |\Sigma|)}$ for any $\delta \leq 1/6, \Sigma$ and $k$. Assume $|\Sigma| = 2^b$   and $b \mid k$ for simplicity. Write ${\bf x} \in \BC^k$ as $(x_{i,j})_{i \in [b], j\in [k/b]}$. Then, for any $a \in [2^{k/b}]$, let $(\cC({\bf x}))_a = (\cH(x_{i,1}, \dots, x_{i, k/b})_a : i \in [b]) \in \BC^b$ where $\cH$ is the classical Hadamard encoding $\cH: \BC^r \to \BC^{2^r}$ defined as $\cH({\bf y}) = (\sum_{i=1}^r y_i \xi_i \pmod 2: \xi_1, \dots, \xi_r \in \BC)$. It is well-known that $\cH$ is a $2$-query LCC, and from this, it is easy to check that $\cC$ is also. The parameters follow directly from the construction. A simple modification of this construction  gives $(2^{O(\delta k/\log |\Sigma|)}/\delta)$-length $2$-query LCCs that tolerate $\delta n$ corruptions. The proof of \cref{thm:main} shows $n \geq \exp(\delta^4 k/\log |\Sigma|)$ which is therefore tight upto $\poly(\delta)$ factors in the exponent.

The 2-query LCC described above is a linear code over $\F_{2^b}$. For linear codes $\cC \subseteq \F_q^n$ (i.e., $\cC$ is a linear subspace of $\F_q^n$), where $q = p^r$ for a prime $p$, \cite{BDSS16} showed that $n \geq \exp(\delta k/r) = \exp(\delta k/\log_p |\Sigma|)$ where $k = \log|\cC|$ is the message length and $|\Sigma| = p^r$. Thus, in terms of dependence on $k$ and $|\Sigma|$, we extend the result of \cite{BDSS16} from linear codes to all zero-error LCCs. Moreover, this work is much more elementary and simple than \cite{BDSS16} which uses non-trivial results from additive combinatorics.

It is important to note that \cref{thm:main} cannot be true for 2-query LDCs. Such a result would contradict the construction in \cite{DG15}  of a zero-error $2$-query LDC  with $\log n = \log |\Sigma| = \exp(\sqrt{\log k}) = k^{o(1)}$ and $\delta = \Omega(1)$. So, our result can be interpreted as giving a separation between zero-error LCCs and LDCs over large alphabet. We conjecture that the zero-error restriction in the theorem can be removed, which if true, would yield the first separation between general LCCs and LDCs.
It is still quite unclear what the correct lower bound for $2$-query LDCs should look like. As mentioned above, Katz and Trevisan \cite{KT00} show that $n \geq \Omega(\delta k^2/\log^2 |\Sigma|)$. And the quantum arguments of \cite{KdW03, WdW05} give the lower bound $n \geq \exp(\delta k/|\Sigma|^2)$ which becomes trivial when $|\Sigma| = \Omega(n)$. 

\ignore{
Let $b=\log(|\Sigma|)$ be the bit-length of the alphabet $\Sigma$. Let $\cC:\BC^k \to \Sigma^n$ be a $(2,\delta, b)$-LDC i.e. for every $i\in [k]$ there exists a matching $\cM_i$ over $[n]$ of size $\delta n$ such that following is true: $x_i$ can be recovered from $C(x)_u, C(x)_v$ for every $\{u,v\}\in \cM_i$. Define $f(n,b,\delta)$ be the largest $k$ for which such an LDC exists. It is proved using quantum arguments that \footnote{see \url{http://arxiv.org/pdf/quant-ph/0403140v2.pdf}} $$f(n,b,\delta)\lesssim 4^b \log(n)/\delta.$$
Also by Katz-Trevisan type random restrictions i.e. a random subset of size about $\sqrt{n/\delta}$ should contain an edge from most of the matchings, we can show $$f(n,b,\delta)\lesssim b\cdot \sqrt{n/\delta}.$$

Clearly both bounds are not optimal, one has a better dependence on $b$ and the other one has a better dependence on $n$. What about the dependence on $\delta$? Even this is not so clear. By concatenating $1/\delta$ Hadamard codes and stacking $b$ on top of each other, we get a code $(2,\delta,b)$-LDC with message length $$k \gtrsim b\log(\delta n)/\delta.$$ This seems like a good guess for the correct bound, but it's not! There exists $(2,1/2,b)$-LDC where $\log n=b=\exp(\sqrt{\log k})=k^{o(1)}$ arising from PIR schemes. These constructions violate such a bound! 

}

\subsection{Proof Overview}
\label{sec:introproof}
Like most prior work on $2$-query LDCs and LCCs, we view the query distribution of the local correcting algorithm as a graph. However, these previous works did not exploit the structure of the graph much beyond its size and degree, whereas our bound is due to a detailed use of the graph structure.

Let $\cC:\BC^k \to \Sigma^n$ be a $2$-query LCC. So, for every $i \in [n]$, there is a corrector algorithm $\cA_i$ that when given access to ${ z} \in \Sigma^n$ with Hamming distance at most $\delta n$ from some codeword ${ y}$, returns $y_i$ with probability at least $2/3$. Assuming non-adaptivity, the algorithm $\cA_i$ chooses its queries from a distribution on $[n]^2$.
Katz and Trevisan \cite{KT00} show how to extract a matching $M_i$ of $\Omega(\delta n)$ disjoint edges on $n$ vertices such that for any edge $e = (j, k)$ in $M_i$, $$\Pr_y\left[\cA_i(y) = y_i \mid \cA \text{ queries $y$ at positions }j\text{ and }k\right] > \frac12 +\eps$$ for some constant $\eps>0$, where the probability is over a uniformly random codeword $y\in \cC$. For zero-error LCCs, the situation is simpler in that essentially, for {\em every} codeword $y$ and edge $e \in M_i$, $\cA_i(y)$ returns $y_i$ when it queries the elements of $e$. This is not exactly correct but let us suppose it's true for the rest of this section.

Let $G$ be the union of $M_1, \dots, M_n$. So, for every edge $(j,k)$ in $G$, there is an $i$ such that $(j,k) \in M_i$. Suppose our goal is to guess an unknown codeword $c$ given the values of a small subset of coordinates of $c$. We assign labels in $\Sigma$ to vertices of $G$ corresponding to the subset of coordinates of $c$ that we know already. Now, imagine a propagation process where we deduce the labels of unlabeled vertices by using the corrector algorithms. For example, if $(j,k) \in M_i$, $j$ and $k$ are labeled but $i$ is not, we can use $\cA_i$ to deduce the label at vertex $i$. Similarly, if $(x,y) \in M_u$ and $(u, v) \in M_w$, and $x, y, v$ are labeled but $u$ and $w$ are not, we can run $\cA_u$ to deduce the label of $u$ and then $\cA_w$ to deduce the label of $w$. The set of labels we infer will be the values of $c$ at the corresponding coordinates. The goal of our analysis is to show that there is a set $S$ of $O_\delta(\log n)$\footnote{$O_\delta(\cdot)$ means that the involved constant can depend on $\delta$. } vertices such that if the labels of $S$ are known, then the propagation process can determine the labels of all $n$ vertices. This immediately implies that the total number of codewords, $2^k$, is at most $|\Sigma|^{|S|}$ and therefore, $k = O_\delta(\log n \cdot \log |\Sigma|)$. Instead, Katz and Trevisan~\cite{KT00} show that if you know the labels of $\sqrt{n}$ uniformly random coordinates, then you can recover the labels of most of the coordinates which leads to the bound $k=O_\delta(\sqrt{n} \cdot \log |\Sigma|)$. Intuitively, their lower bound is just one step of the propagation process.

The propagation process is perhaps more naturally described on a (directed) $3$-uniform hypergraph where there is an edge $(i,j,k)$ if $(j,k) \in M_i$. It ``captures'' $i$ if $(i,j,k)$ is an edge and $j, k$ are already captured. Coja-Oghlan et al.~\cite{COW12} study exactly this process on random undirected $3$-uniform hypergraphs in the context of constraint satisfaction problem solvers. Unfortunately, their techniques are specialized to random hypergraphs. The propagation process is also related to hypergraph peeling \cite{MT12, MW15}, but again, most theoretical work is limited to random hypergraphs.

To motivate our approach, suppose $M_1, \dots, M_n$ are each a perfect matching. For a set $S \subseteq [n]$, let $R(S)$ denote the set of vertices to which we can propagate starting from $S$. If $R(S) = [n]$, we are done. Otherwise, we show that we can double $|R(S)|$ by adding one more vertex to $S$. Note that for any $i \notin R(S)$, no edge in $M_i$ can lie entirely inside $R(S)$, for then, $i$ would also have been reached. So, each vertex in $R(S)$ must be incident to one edge in $M_i$ for every $i \notin R(S)$. This makes the total number of edges between $R(S)$ and $[n]\setminus R(S)$ belonging to $M_i$ for some $i \not \in R(S)$ equal to $|R(S)| \cdot (n-|R(S)|)$. By averaging, there must be $j \notin R(S)$ that is incident to at least $|R(S)|$ edges, each belonging to some $M_i$ for $i \notin R(S)$. Moreover, all these $|R(S)|$ edges must belong to matchings of different vertices. Hence, adding $j$ to $S$ doubles the size of $R(S)$. Hence, for some $S$ of size $O(\log n)$, $R(S) = [n]$.

In the above special case (where all the matchings were perfect), we used the fact that the size of the cut between $R(S)$ and the rest of the graph is large and that many of these edges belong to  $M_i$ for $i \not \in R(S)$. We observe that for any graph obtained from an LCC as above, this situation exists whenever $R(S)$ is not too large already and the minimum degree of every vertex in the graph is large (say, $\poly(\delta) \cdot n$). This is because each vertex in $R(S)$ will be incident to many edges in matchings $M_i$ for $i \notin R(S)$ (using the minimum degree requirement and that $|R(S)|$ is small) and such edges cannot have both endpoints inside $R(S)$ (as then $i \in R(S)$). So, indeed, there will be many edges with labels not in $R(S)$ crossing the cut, and averaging will yield a vertex whose addition to $S$ will make $R(S)$ grow by a multiplicative factor. Therefore, if the minimum degree requirement is met, we can keep repeating this process until $R(S)$ becomes large, of size $\poly(\delta)\cdot n$.
Now, in a key lemma of our proof, we show that for any graph obtained from an LCC as above, we can greedily find a subset of the vertices $V'$ such that the the subgraph induced by the vertices of $V'$ and the edges labeled by $V'$ has large minimum degree. So, we can repeatedly apply the above argument to $V'$ to find a subset $S$ of size $O_{\delta}(\log n)$ such that $R(S)$ contains $\poly(\delta)\cdot n$ vertices.

Recall that our goal is to find a small set $S$ such that $R(S)=[n]$. So, at this stage, we would ideally like to continue the argument on $V'' = [n]\setminus R(S)$. The only issue we can face is that the graph on $V''$ restricted to edges labeled by $V''$ may not have the LCC structure. Indeed, it could be that most edges labeled by $V''$ are not spanned by vertices in $V''$. However in this case, there will be a vertex $u$ in $V''$ incident to many $V''$-labeled edges that have their other endpoints in $R(S)$, so that we can increase $R(S)$ by adding $u$ to $S$. Thus, either $R(S)$ may be grown directly or else the rest of the vertices looks approximately like an LCC, so that we can recurse. Modulo some important technical details, our proof is now complete\footnote{An earlier version \cite{BGprelim} of this paper had a different argument for the main theorem, based on a ``decomposition theorem'' proved using the Szem\'eredi regularity lemma for directed graphs \cite{Szem78, AS04}. The idea was to partition the graph into a constant number of edge expanders. In each such part, the sizes of cuts are large and so the propagation process can be easily analyzed. The proof given here is simpler and yields much better dependence on $\delta$. However, because the decomposition theorem for directed graphs may be of general interest, we have included it in Appendix~\ref{sec:decomposition} of this paper.}.

The zero-error assumption seems necessary to make the propagation process well-defined. Otherwise, for each labeled vertex, there is some probability that the label is incorrect for the codeword in question. But since there may be $\Omega(\log n) = \omega(1)$ steps of propagation, the error probability may blow up by this factor. So, it seems we need different techniques to handle correctors that have constant probability of error when the input is a codeword. One possibility is using information theory to better handle the spread of error\footnote{This approach is taken in \cite{Jain06} to prove an exponential lower bound for smooth 2-query LDCs over binary alphabet when the decoder has subconstant error probability. Jain's analysis seems to work only for binary codes but is similar in spirit to ours.}. 

\section{Zero-error 2-query LCCs}\label{sec:prelim}
We begin by formally defining zero-error $2$-query LCCs.

\begin{definition}
Let $\Sigma$ be some finite alphabet and let $\cC\subset \Sigma^n$ be a set of codewords. $\cC$ is called a $(2,\tau)$-LCC with zero-error if there exists a randomized algorithm $\cA$ such that following is true:
\begin{enumerate}
\item $\cA$ is given oracle access to some $z\in \Sigma^n$ and an input $i\in [n]$. It outputs a symbol in $\Sigma$ after making at most $2$ non-adaptive queries to $z$.
\item  If $z\in \Sigma^n$ is $\tau$-close to some codeword $c\in \cC$ in Hamming distance, then for every $i\in [n]$, $\Pr[\cA^z(i)=c_i]\ge 2/3.$
\item If $c\in \cC$, then for every $i\in [n]$, $\Pr[\cA^c(i)=c_i]=1$ i.e. if the received word has no errors, then the local correction algorithm will not make any error.
\end{enumerate}
\end{definition}
Note that the above definition differs from the standard notion of non-adaptive 2-query LCCs only in part (3) above. The choice of $2/3$ in part (2) of the definition above is somewhat arbitrary. We can make it any constant greater than $1/2$. More generally, it is only required that for every $\sigma \neq c_i, \Pr[\cA^z(i) = c_i] > \Pr[\cA^z(i) = \sigma] + \eps$ for some $\eps > 0$, i.e., $c_i$ should win the plurality vote among all symbols by a constant margin.

We next show that the corrector for any zero-error LCC can be brought into a ``normal" form. A similar statement is known for general LDCs and LCCs \cite{KT00, Yek11} but we need to be a bit more careful because we want to preserve the zero-error property. Note that the proof overview in \cref{sec:introproof} assumed that the set $T_1$ below is empty.


\begin{lemma}\label{lem:lcc_matching}
Let $\cC\subset \Sigma^n$ be a $(2,\tau)$-LCC with zero error. Then, there exists a partition of $[n]=T_1\cup T_2$ such that:

\begin{enumerate}
\item For every $i\in T_1$, there exists a distribution $\cD_i$ over $[n]\cup \{\phi\}$ and algorithms $\cR^i_j$ for every $j\in [n]\cup\{\phi\}$ such that for every codeword $c\in \cC$, $$\Pr_{j\sim \cD_i}\left[\cR^i_j(c_j)=c_i\right]\ge \frac{2}{3}.\footnote{Here $c_\phi$ is an empty input defined for ease of notation.}$$ Moreover the distribution $\cD_i$ is smooth over $[n]$ i.e. for every $j\in [n]$, $\Pr_{\cD_i}[j]\le \frac{4}{\tau n}$.

\item For every $i\in T_2$, there exists a matching $\cM_i$ of edges in $[n]\setminus \{i\}$ of size $|\cM_i|\ge \frac{\tau}{4}n$ such that: For every $c\in \cC$,  $c_i$ can be recovered from $(c_j,c_k)$ for any $(j,k) \in \cM_i$ i.e. there exists algorithms $\cR^i_{j,k}$ for every edge $(j,k)\in \cM_i$ such that for every $c\in \cC$, $$\cR^i_{j,k}(c_j,c_k)=c_i.$$

\end{enumerate}
\end{lemma}
\begin{proof}
Fix $\eps=\tau/4$. Let $\cA$ be the local corrector algorithm for $\cC$ and let $\cQ_i$ be the distribution over 2-tuples of $[n]$ corresponding to the queries $\cA(i)$ makes to correct coordinate $i$.\footnote{Wlog, we can assume $\cA(i)$ always queries two coordinates.} Let $\supp(\cQ_i)$ be the set of edges in the support of $\cQ_i$. We have two cases:\\
\textbf{Case 1:} $\supp(\cQ_i)$ contains a matching of size $\eps n$.\\
In this case, we include $i\in T_2$ and define $\cM_i$ to be a matching of size $\eps n$ in $\supp(\cQ_i)$. Let $\cR^i_{j,k}(z_j,z_k)$ be the output\footnote{Note that $\cR^i_{j,k}$ might use additional randomness.} of $\cA^z(i)$ when it samples $(j,k)$ from the distribution $\cQ_i$. So we have for every $\sigma\in \Sigma$, $$\Pr_{(j,k)\sim \cQ_i}[\cR^i_{j,k}(z_j,z_k)=\sigma]=\Pr[\cA^z(i)=\sigma].$$  Now since our LCC is zero-error, for every $(j,k)\in \supp(\cQ_i)$, we have $\cR^i_{j,k}(c_j,c_k)=c_i$. This takes care of part (2).\\
\textbf{Case 2:} $\supp(\cQ_i)$ doesn't contain a matching of size $\eps n$.\\
In this case we include $i\in T_1$. Since $\supp(\cQ_i)$ doesn't contain a matching of size $\eps n$, there exists a vertex cover of size at most $2\eps n$, say $V_i$. 
Also define $B_i\subset [n]$ to be the set of vertices which are queried with high probability by $\cA^z(i)$ i.e. $$B_i=\left\{j: \Pr[\cA^z(i)\ \text{queries}\ j]\ge \frac{1}{\eps n}\right\}.$$ Clearly $|B_i|\le 2\eps n$ because $\cA^z(i)$ makes at most two queries.
 
 We now define a new one-query corrector for $i$, $\tilde{\cA^z}(i)$ as follows: simulate $\cA^z(i)$, but whenever $\cA^z(i)$ queries $z$ at a coordinate in $V_i\cup B_i $, $\tilde{\cA^z}(i)$ doesn't query that coordinate and assumes that the queried coordinate is $0$ (or some fixed symbol in $\Sigma$).
Note that $\tilde{\cA^z}(i)$ makes at most one query to $z$ since $V_i$ is a vertex cover for the support of $\cQ_i$.
 Also $\tilde{\cA^c}(i)$ behaves exactly like $\cA^{c'}(i)$ where ${c'}$ is the word formed by zeroing out the $V_i\cup B_i$ coordinates of $c$. Since $|V_i\cup B_i|\le 4\eps n\le \tau n$, we have $$\Pr[\tilde{\cA^c}(i)=c_i]=\Pr[\cA^{c'}(i)=c_i]\ge \frac{2}{3}.$$ Now define the distribution $\cD_i$ over $[n]\cup \{\phi\}$ as: $$\Pr_{\cD_i}[j]=\Pr[\tilde{\cA^z}(i)\ \text{queries}\ j]$$ for $j\in [n]$ and $$\Pr_{\cD_i}[\phi]=\Pr[\tilde{\cA^z}(i)\ \text{doesn't make any query}].$$ Since we never query elements of $B_i$, we have the required smoothness i.e. $\Pr_{\cD_i}[j]\le 1/(\eps n)$ for all $j\in [n]$. Also define $\cR^i_j(z_j)$ to be the output (can be randomized) of $\tilde{\cA^z}(i)$ when it queries $j\in [n]$ and $\cR^i_\phi(c_\phi)$ to be the output (can be randomized) of $\tilde{\cA^z}(i)$ when it doesn't make any query where $c_\phi$ is an empty input defined for ease of notation. By definition, we have $$\Pr_{j\sim\cD_i}[\cR^i_j(c_j)=c_i]=\Pr[\tilde{\cA^c}(i)=c_i]\ge \frac{2}{3}.$$ This proves part (1).
\end{proof}
\section{Proof of lower bound}\label{sec:lowerbound}
\def\({\left(}
\def\){\right)}
\subsection{An information theoretic lemma}
The proof of Theorem~\ref{thm:main} works by showing that there is randomized algorithm which can guess an unknown codeword $c\in \cC\subset \Sigma^n$ with high probability by making a small number of queries. From this we would like to show that $|\cC|$ cannot be large. We will apply Fano's inequality which is a basic information theoretic inequality to achieve this. We will assume familiarity with basic notions in information theory; we refer the reader to \cite{CT12} for precise definitions and the proofs of the facts we use. Given random variables $X,Y,Z$, let $H(X)$ be the entropy of $X$ which is the amount of information contained in $X$. $H(X|Y)$ is the conditional entropy of $X$ given $Y$ which is the amount of information left in $X$ if we know $Y$. The mutual information $I(X;Y)=H(X)-H(X|Y)=H(Y)-H(Y|X)$ is the amount of common information between $X,Y$. If $X,Y$ are independent, then $I(X;Y)=0$. The conditional mutual information $I(X;Y|Z)$ is the mutual information between $X,Y$ given $Z$. We have the following chain rule for mutual information: $$I(X;YZ)=I(X;Z)+I(X;Y|Z).$$ We also need the following basic inequality: $$I(X;Y|Z)\le H(X|Z) \le \log |\cX|$$ where $\cX$ is the support of the random variable $X$. We will now state Fano's inequality which says that if we can predict $X$ very well from $Y$ i.e. there is a predictor $\hat{X}(Y)$ such that $\Pr[\hat{X}(Y)\ne X]\le p_e$ where $p_e$ is small, then $H(X|Y)$ should be small as well (see \cite{CT12} for a proof). More precisely,
\begin{align*}
H(X|Y)\le h(p_e)+p_e\log(|\cX|-1) \tag{Fano's inequality}
\end{align*}
 where $h(x)=-x\log x-(1-x)\log(1-x)$ is the binary entropy function and $\cX$ is the support of random variable $X$.
\begin{lemma}\label{lem:oracleguess}
Suppose there exists a randomized algorithm $\cP$ such that for every $c\in \cC\subset \Sigma^n$, given oracle access to $c$, $\cP$ makes at most $t$ queries to $c$ and outputs $c$ with probability $\ge 1/2$, then $\log |\cC|\le O(t\log|\Sigma|).$
\end{lemma}
\begin{proof}
Let $X$ be a random variable which is uniformly distributed over $\cC$. Let $R$ be the random variable corresponding to the random string of  the algorithm $\cP$ and let $S(R)$ be the set of coordinates queried by $\cP$ when the random string is $R$. We can guess the value of $X$ with probability $\ge 1/2$ given $X_{S(R)},R$ where $X_{S(R)}$ is the restriction of $X$ to $S(R)$. By Fano's inequality, $$H(X\ \vert\ X_{S(R)},R) \le h(1/2)+\frac{1}{2}\cdot \log(|\cC|-1)\le 1+ \frac{1}{2}\log|\cC|.$$
We can bound the mutual information between $X$ and $X_{S(R),R}$ as follows:
\begin{align*}
I(X;X_{S(R)},R) &= I(X;R)+I(X;X_{S(R)}|R) \tag{Chain rule for mutual information.}\\
&\le 0+H(X_{S(R)}|R) \tag{Since $X$ and $R$ are independent.}\\
&\le t\log|\Sigma|.
\end{align*}
But we also have $$I(X;X_{S(R)},R)= H(X)-H(X|X_{S(R)},R) \ge \log |\cC| - \frac{1}{2}\log |\cC|-1\ge \frac{1}{2}\log|\cC| - 1.$$
Combining the upper and lower bound for $I(X;X_{S(R)},R)$, we get the required bound. 
\end{proof}

\subsection{Proof of Theorem~\ref{thm:main}}
The following is a restatement of Theorem~\ref{thm:main}.
\begin{theorem}\label{thm:2lcc_lowerbound_directed}
Let $\cC\subset \Sigma^n$ be a $(2,\tau)$-LCC which is zero-error, then $|\cC|\le \exp\left(O(\tfrac{1}{\tau^4} \cdot \log n \cdot\log{|\Sigma|})\right)$.
\end{theorem}
\begin{proof}

We will construct a randomized algorithm $\cP$ such that for every $c\in \cC $, given oracle access to $c$, $\cP$ makes at most $O(\tfrac{1}{\tau^4} \cdot \log n)$ queries to $c$ and outputs $c$ with probability $\ge 1-1/n$. By Lemma~\ref{lem:oracleguess}, we get the required bound.

Let $[n]=T_1\cup T_2$ be partition of coordinates given by Lemma~\ref{lem:lcc_matching}. 

\begin{claim}\label{claim:recovering_T1}
Algorithm $\cP$ can learn $c|_{T_1}$ with probability $\ge 1-1/n$ by querying a uniformly random (sampled with repetitions) subset  $S$ of size $r=O(\tfrac{1}{\tau^2} \cdot \log n)$.
\end{claim}
\begin{proof}
 Let $S=\{Z_1,\cdots, Z_r\}$ where each $Z_i$ is a uniformly random element of $[n]$. By Lemma~\ref{lem:lcc_matching}, for every $u\in T_1$, we have a smooth distribution $\cD_u$ over $[n]$ and algorithms $\cR^u_v$ for every $v\in [n]$. Let's fix $u\in T_1$ and let $p_v=\Pr_{\cD_u}[v]$. By smoothness, $p_v\le \frac{4}{\tau n}$ for every $v\in [n]$.
 The algorithm $\cP$ estimates $c_u$ as follows: Define the weight of $\sigma$ to be $$W_\sigma = p_\phi\cdot \Pr[\cR^u_\phi=\sigma]+\frac{1}{r}\sum_{i=1}^r np_{Z_i}\cdot \Pr[\cR^u_{Z_i}(c_{Z_i})=\sigma]$$ and output the symbol with the maximum weight. We will show that $$\Pr[\cP\ \text{guesses}\ c_u\ \text{incorrectly}]\le \frac{1}{n^2}.$$ 
For $\sigma\in \Sigma$ and $v\in [n]\cup\{\phi\}$, let $f^\sigma_v=\Pr[\cR^u_v(c_v)=\sigma]$. The weight of $\sigma$ is given by $$W_{\sigma}=p_\phi f^\sigma_\phi+\frac{1}{r}\sum_{i=1}^r np_{Z_i}f^\sigma_{Z_i}.$$ We can calculate the expected value of the weight as $$\E[W_\sigma]=p_\phi f^\sigma_\phi+\E[np_{Z_1}f^\sigma_{Z_1}]=p_\phi \Pr[\cR^u_\phi(c_\phi)=\sigma]+\sum_{v\in [n]} p_v \Pr[\cR^u_v(c_v)=\sigma]=\Pr_{v\sim \cD_u}[\cR^u_v(c_v)=\sigma].$$ Therefore $W_{\sigma}$ is an unbiased estimator for $\Pr_{v\sim \cD_u}[\cR^u_v(c_v)=\sigma]$. 
Also $p_{Z_i}\le \frac{4}{\tau n}$ and $f^\sigma_{Z_i}\le 1$, so $np_{Z_i}f^\sigma_{Z_i}\le \frac{4}{\tau}.$ Applying Hoeffding's inequality, 
$$\Pr\left[|W_\sigma-\E[W_\sigma]|\ge \frac{1}{20} \right]\le \exp\left(-\Omega(r\tau^2)\right)\le 1/2n^2$$ when $r\gg \frac{1}{\tau^2}\log n$.
By Lemma~\ref{lem:lcc_matching}, $$\E[W_{c_u}]=\Pr_{v\sim \cD_u}[\cR^u_v(c_v)=c_u]\ge \frac{2}{3}.$$ Therefore, $\Pr[W_{c_u}\le \frac{2}{3}-\frac{1}{20}]\le 1/2n^2$.
Now we will show that no other symbol can have higher weight than $W_{c_u}$ except with probability $\frac{1}{2n^2}$. For this let us look at 
\begin{align*}
\sum_{\sigma\in \Sigma} W_\sigma &=\sum_{\sigma}p_\phi f^\sigma_\phi+ \frac{1}{r}\sum_{i=1}^r np_{Z_i}\sum_{\sigma}f^\sigma_{Z_i}\\
&=p_\phi \sum_\sigma \Pr[\cR^u_\phi=\sigma]+\frac{1}{r}\sum_{i=1}^r np_{Z_i}\sum_{\sigma}\Pr[\cR^u_{Z_i}(c_{Z_i})=\sigma]\\
&=p_\phi+\frac{1}{r}\sum_{i=1}^r np_{Z_i}
\end{align*}

So $\E[\sum_{\sigma\in \Sigma} W_\sigma]=p_\phi+\E[np_{Z_1}]=1$ and $np_{Z_i}\le \frac{4}{\tau}$. Therefore by Hoeffding's inequality applied again, we get $$\Pr\left[\left|\sum_{\sigma\in \Sigma} W_\sigma -1\right|\ge \frac{1}{20}\right] \le \exp\left(-\Omega(r\tau^2)\right)\le \frac{1}{2n^2}$$ when $r\gg \frac{1}{\tau^2}\log n$.
So with probability $\ge 1-\frac{1}{n^2}$, we have $W_{c_u}\ge \frac{2}{3}-\frac{1}{20}$ and $\sum_{\sigma\in \Sigma} W_\sigma \le 1+\frac{1}{20}$. Therefore with probability $\ge 1-\frac{1}{n^2}$, $c_u$ will be the symbol with maximum weight and the algorithm $\cP$ will guess $c_u$ correctly with probability $\ge 1-\frac{1}{n^2}$. By union bound, we get that $\cP$ can guess $c_u$ correctly for all $u\in T_1$ with probability $\ge 1-\frac{1}{n}$.
\end{proof}

We will now show that after learning $c|_{T_1}$, $\cP$ can now learn $c|_{T_2}$ by querying a further $O_\tau(\log n)$ coordinates from $c$ and this process will be deterministic i.e. no further randomness is needed. Define $R(S)$ to be the set of coordinates of $c$ that can be recovered correctly given $c|_S$. In Claim~\ref{claim:recovering_T1}, we have shown that if $S$ is a randomly chosen subset of size $O_\tau(\log n)$, then $T_1\subseteq R(S)$ with probability $\ge 1-\frac{1}{n}$. From now on we assume that $\cP$ has already recovered coordinates of $T_1$ correctly i.e. $T_1\subseteq R(S)$.
If $T_2\subseteq R(S)$ then we are done, the algorithm $\cP$ can output the entire $c$ with probability $\ge 1-\frac{1}{n}$. So we can assume that $T_2\nsubseteq R(S)$. Our goal is to show that we can add a further $O(\poly(1/\tau) \cdot \log n)$ vertices to $S$ and have $R(S)=V=T_1\cup T_2$. We show that this is indeed the case in the next section by proving the following claim, which completes the proof.
\begin{claim}\label{claim:graph theoretic}
There exists a set $S$ of size $O((1/\tau)^4 \cdot \log n)$ such that $R(S \cup T_1) = V$.	
\end{claim}
\renewcommand{\qedsymbol}{\sc q.e.d.}
\end{proof}

\subsection{Proof of Claim~\ref{claim:graph theoretic}}\label{sec:graph}
Claim~\ref{claim:graph theoretic} is purely graph theoretical. 
Let $G=(V,E)$ be the graph with $V=[n]=T_1\cup T_2$ and $E=\cup_{i\in T_2} \cM_i$ where $\cM_i$ are partial matchings of size at least $(\tau/4) n$ given by Lemma~\ref{lem:lcc_matching}. Let $\delta:=\tau/4$. We will label each edge in $E$ with a label in $T_2$ indicating which matching it belongs to. We can have parallel edges in $E$, but they will have different labels since they belong to different matchings. Recall that $R(S)$ is the set of coordinates of $c$ that can be inferred from $c|_S$. Lemma~\ref{lem:lcc_matching} implies the following closure property for $R(S)$: if $(i,j)\in \cM_k$ and $i,j\in R(S)$ then $k\in R(S)$.
Next, we define $R(S)$  formally based on the graph $G$ using this closure property.
\begin{definition}
Let $G = (V,E)$ as above. Let $S\subseteq V$.
We define the set $R_G(S) \subseteq V$ to be the smallest set of vertices such that:
\begin{enumerate}
	\item $S \subseteq R_G(S)$
	\item For all $i,j \in R_G(S)$ and $k\in [n]$, if $(i,j) \in \cM_k$, then $k \in R_G(S)$. (In words, if there exists an edge $(i,j)$ in the graph $G$ labeled with $k$ and both $i$ and $j$ are in $R_G(S)$, then so is $k$.) 
\end{enumerate}
\end{definition}
	
(When the context is clear, we will use $R(S)$ instead of $R_G(S)$.)
Our goal is to show that in any graph $G$ as above, there exists a set $S\subseteq V$ of size $\poly(1/\delta) \cdot \log(n)$ such that $R_G(S \cup T_1) = V$.
As a first step, we get rid of the set $T_1$, by showing that proving the claim in the case $T_1 = \emptyset$ implies Claim~\ref{claim:graph theoretic} for any other set.
To see that observe that if we take $G'$ to be the union of $G$  with a collection of partial matching $\{\cM_{j}\}_{j\in T_1}$, then $R_{G'}(S) \subseteq R_{G}(S \cup T_1)$ for any set $S \subseteq V$.
Thus, it suffices to introduce dummy matchings $\{\cM_{j}\}_{j\in T_1}$ for each $\cM_j$ of  size $\delta n$, and prove that there exists a set $S$ of size 
$\poly(1/\delta) \cdot \log(n)$ such that $R_{G'}(S) = V$.

\begin{claim}[Claim~\ref{claim:graph theoretic}, case $T_1=\emptyset$, restated]\label{claim:T1 empty}
Let $G = (V,E)$ be a graph with $V=[n]$ and $E = \cM_1 \cup \cdots \cup \cM_n$ where each $\cM_i$ is a partial matching of size at least $\delta n$. Then, there exists a subset $S \subseteq V$ of size $O((1/\delta)^4 \cdot \log n)$  such that $R_G(S) = V$.
\end{claim}

From here henceforth we assume (without loss of generality) that $T_1 = \emptyset$ and $T_2 = [n]$, and prove Claim~\ref{claim:T1 empty}.
The following lemma tells us that we can find a subgraph $G'$ of $G$ such that each vertex in $G'$ has high degree. Note that the lemma finds a subgraph restricted to a set of vertices $V'$, and also restricted to the set of edges labeled with $V'$. 

We shall use this lemma inductively. 
During induction, we will remove some edges from the matchings. Thus, instead of asserting that all matchings are of size at least $\delta |V|$, we assume that all but $0.1 \delta |V|$ of the matchings have at least $0.9 \delta |V|$ edges.

\begin{lemma}[Clean-Up Lemma]\label{lemma:cleanup2}
	Let $G = (V,E)$ be a graph with a finite set of vertices $V$ and $E = \bigcup_{i\in V}\cM_i$, where each $\cM_i$ is a partial matching on $V$. Assume all but $0.1 \delta |V|$  of the matchings $\cM_i$ have size at least $0.9 \delta |V|$.
	Then, there exists a subset $V' \subseteq V$ of size at least $\delta \cdot |V|$ so that
	the graph $G' = (V',E')$ where $E' = \bigcup_{i\in V'} {\cM_i \cap (V' \times V')}$ has minimal degree at least $(\delta^2/4) \cdot |V|$.
\end{lemma}
\begin{proof}
We find the set $V'$ greedily. Let $\delta' := \delta^2/4$. Initialize $V' = V$.
If the minimum degree in the remaining graph on $V'$ is at least $\delta' \cdot |V|$ then we stop.
Otherwise, remove the vertex $i\in V'$ with minimal degree, and remove all edges labeled $i$.
We repeat this process until no vertices of degree smaller than $\delta' \cdot |V|$ exist.

If the process stopped when $|V'| \ge \delta |V|$ then we are done.
We are left to show that the process cannot proceed past this point.
Let's assume by contradiction that we can continue the process after this point. As we decrease the size of $V'$ by one in each iteration, we must reach at a certain point of the process to a set of vertices  $V' = V^{*}$ of size exactly $\delta |V|$. 
Denote by $$E^{*}(V') := \bigcup_{i\in V^{*}} {\cM_i \cap (V' \times V')}.$$
Next, we upper and lower bound $|E^{*}(V^*)|$ to derive a contradiction.

The upper bound  $|E^{*}(V^{*})| \le |V^{*}|\cdot |V^{*}|/2$ follows since the edges $E^{*}(V^{*})$ form a collection of $|V^{*}|$ partial matchings on $V^{*}$. 
To lower bound $|E^{*}(V^{*})|$ we use the properties of the greedy process. 
The initial size of the set $E^{*}(V')$ (when $V' = V$) is at least $0.9 \delta |V| \cdot (|V^{*}| - 0.1 \delta |V|) \ge 0.9^2 \delta^2\cdot |V|^2$. In every iteration, we remove at most $\delta' |V|$ edges from this set of edges. As there are at most $|V|$ steps, we are left with at least $0.9^2 \delta^2 |V|^2 - \delta' |V|^2$ edges, i.e., $|E^{*}(V^{*})| \ge 0.9^2\delta^2 |V|^2 - \delta' |V|^2$.
Combining both upper and lower bounds on $|E^{*}(V^{*})|$ gives
\[
\frac{1}{2} \cdot \delta^2 \cdot |V|^2 \ge |E^{*}(V^{*})| \ge (0.9^2 \delta^2 - \delta') \cdot |V|^2 = (0.9^2\delta^2 - \delta^2/4) \cdot |V|^2
\]
which yields a contradiction since $1/2 < 0.9^2 - 1/4$.
\end{proof}

\begin{lemma}[Exponentially growing a set of known coordinates]\label{lemma:exp growth}
	Let $G = (V,E)$ be a graph with $V$ and $E = \bigcup_{i\in V} \cM_i$  such that each $v\in V$ has degree at least $d$.
	Then, there exists a subset $S \subseteq V$ of size at most $O((|V|/d) \cdot \log |V|)$ with $|R(S)|\ge d/2$.
\end{lemma}
\begin{proof}
		
We pick the set $S \subseteq V$ iteratively, picking one element in each step. We start with $S = \{v\}$ for some arbitrary $v\in V$.

Assume we picked $t$ elements so far for the set $S$.
If $|R(S)| \ge d/2$, then we are done.
Otherwise, by the definition of $R(S)$, for any $i \in V \setminus R(S)$, none of the edges in the matching $\cM_i$ is inside $R(S)$.
We wish to show that there exists an  $i \in V \setminus R(S)$ with many edges into $R(S)$ marked with labels outside $R(S)$.
Then, we will add $i$ to $S$, which will reveal a lot of new coordinates.

For two disjoint sets of vertices $A,B \subseteq V$ we denote by $E(A,B)$ the set of edges between $A$ and $B$ in the graph $G$. If $A$ consists of one element, i.e., $A=\{a\}$ we denote $E(a,B) = E(A,B)$.
Let $A = R(S)$. Let $B = V \setminus A$. We have
\begin{equation}\label{eq:Edges between A and B}\left|E(A,B) \cap \bigcup_{i\in B} \cM_{i}\right| = \sum_{a\in A}\left|E(a,B) \cap \bigcup_{i\in B} \cM_{i}\right|
= \sum_{a\in A}\left|E(a,V \setminus \{a\}) \cap \bigcup_{i\in B} \cM_{i}\right|
\end{equation}
where the last equality follows since there are no edges labeled $i \in B$ between any two vertices in $A$.
For each $a \in A$ there are at least $d$ edges touching $a$ and at most $|A|$ of them appeared in $\bigcup_{i\in A}\cM_i$, hence 
$
\left|E(a,V \setminus \{a\}) \cap \bigcup_{i\in B} \cM_{i}\right| \ge d - |A| \ge d/2.
$
Plugging this estimate to Eq.~\eqref{eq:Edges between A and B} gives 
\[
\left|E(A,B) \cap \bigcup_{i\in B} \cM_{i}\right| \ge  |A| \cdot d/2\;.
\]
By averaging there exists a vertex $b\in B$ with at least $|A| \cdot  \frac{d}{2|V|}$ edges to $A$ labeled with $B$.
So as long as $|A|=|R(S)| \le d/2$ we are extending the set  $R(S)$ by at least $|R(S)| \cdot\frac{d}{2|V|}$ elements, i.e. by a multiplicative factor of $(1+\frac{d}{2|V|})$. Hence, after $t$ iterations, either $|R(S)| \ge (1+\frac{d}{2|V|})^t$ or $|R(S)| \ge d/2$.
Taking $t = O(\frac{|V|}{d} \cdot \log|V|)$ gives that after at most $t$ iterations $|R(S)| \ge d/2$. 
\end{proof}

\begin{lemma}[Covering $1-\delta$ fraction of the coordinates implies covering all coordinates]\label{lemma:too big is everything}
	Let $G = (V,E)$ be a graph with $V=[n]$ and $E = \cM_1 \cup \cM_2 \cup \ldots \cup \cM_n$ and each $\cM_i$ is a partial matching of size at least $\delta n$.
	Let $S \subseteq V$. If $|R(S)| > (1-\delta)n$, then $R(S) = V$.
\end{lemma}
\begin{proof}
	Let $v \in V$. We show that there is an edge inside $R(S)$ marked $v$. Indeed, there are at least $\delta n$ edges labeled $v$ and they form a partial matching. If $|V \setminus R(S)| < \delta n$, one of these edges do not touch $(V\setminus R(S))$, i.e., it is an edge connecting two vertices in $R(S)$.
\end{proof}

\begin{lemma}[Two Cases]\label{lemma:two cases}
	Let $G = (V,E)$ be a graph with $V=[n]$ and $E = \cM_1 \cup \cM_2 \cup \ldots \cup \cM_n$ where each $\cM_i$ is a partial matching of size at least $\delta n$.
	Let $S \subseteq V$. Assume $|R(S)| \le (1-\delta)n$.
	Then, either
	\begin{enumerate}
		\item There exists an $i \in V\setminus R(S)$ such that $|R(S \cup \{i\})| \ge |R(S)| + 0.01 \cdot \delta^2 \cdot n$.
		\item In the graph $G' = (V',E')$ with $V' = V \setminus R(S)$ and $E'=\bigcup_{i\in V'} {\cM_i \cap (V' \times V')}$ all but at most $0.1 \delta \cdot |V'|$ of the matchings have at least $0.9\delta\cdot n$ edges.
	\end{enumerate}
	\end{lemma}
\begin{proof}
Recall that the labels of edges incident to any vertex $i$ are distinct, since the graph is a union of partial matchings.
Denote by $A = R(S)$ and $B = V \setminus R(S)$.
	Assume for any $i \in B$ there are at most $0.01 \delta^2 \cdot n$ edges to $A$ labeled with labels in $B$. (Otherwise, extend $S$ by $i$ and get $|R(S \cup \{i\})| \ge |R(S)| + 0.01 \delta^2 \cdot n$.)
	Then, there are at most $0.01 \delta^2 \cdot n \cdot |B|$ edges in the cut $(A,B)$ with labels in $B$. By definition of $A=R(S)$, there are no edges between $A$ and $A$ labeled with $B$. 
	Thus, at most $0.01 \delta^2 n \cdot |B|$ edges are missing from the matchings labeled by $B$ if we restrict to edges between $B$ and $B$. Hence, at most $0.1 \delta \cdot |B|$ of the matchings may miss more than $0.1 \delta \cdot n$ of their edges.
\end{proof}

We are now ready to prove Claim \ref{claim:T1 empty}.

\begin{proof}[Proof of Claim	 \ref{claim:T1 empty}]
	Initialize $S := \emptyset$.
We repeat the following process. 	While $R(S) \neq V$, check if there exists $i \in V \setminus R(S)$ such that $|R(S \cup \{i\})| \ge |R(S)| + 0.01 \delta^2 n$. 
	We have two cases:
	\begin{enumerate}
		\item If such an $i$ exists, update $S := S \cup \{i\}$.
		\item Else, let $G' = (V',E')$ where $V' =  V \setminus R(S)$ and $E'=\bigcup_{i\in V'} {\cM_i \cap (V' \times V')}$. Let $M'_i := \cM_i \cap (V' \times V')$. 
	By Lemma~\ref{lemma:too big is everything}, $|V'| \ge \delta n$.
	By Lemma~\ref{lemma:two cases}, all but at most $0.1 \delta |V'|$ of the matchings $M'_i$ for $i\in V'$ have at least $0.9 \delta n$ edges.
	Denote by $\delta' = 0.9 \delta n/|V'| \ge \delta$.
	We apply Lemma~\ref{lemma:cleanup2} on $G'$ to get  a subgraph $G'' = (V'',E'')$ defined by a subset $V''$ of size $\Omega(\delta'|V'|)$ and $E'' = \bigcup_{i\in V''}{\cM_i \cap (V'' \times V'')}$ with minimal degree $d  = \Omega((\delta')^2 \cdot |V'|) \ge \Omega(\delta^2 n)$.
	We apply Lemma~\ref{lemma:exp growth} on $G''$ to get a set $S''\subseteq V''$ of size $O(\log |V''| \cdot (|V''|/d)) = O(\log n \cdot (1/\delta')^2)$ with
	$|R_{G''}(S'')| \ge \Omega(d) \ge \Omega(\delta^2 n)$. We update $S := S \cup S''$.

	\end{enumerate}

	The number of times we apply case 1 or case 2 is at most $O(1/\delta^2)$, since each such step introduces $\Omega(\delta^2 n)$ new vertices to $R(S)$.
	In each application of case 2, at most $O((1/\delta')^2 \cdot \log n) \le O((1/\delta^2) \cdot \log n)$ elements are added to $S$. Overall, the size of $S$ at the end of the process will be 
	\[
	O\(\tfrac{1}{\delta^2}\) + O\(\tfrac{1}{\delta^2} \cdot \tfrac{1}{\delta^2} \cdot \log n\)  = 
	O\( \tfrac{1}{\delta^4} \cdot \log n\)\;.\qedhere
	\]
	\end{proof}

\bibliographystyle{alpha}
\bibliography{references}

\appendix

\section{LDCs from LCCs}\label{app:LCC vs LDC}

In this section, we will show that $q$-query LCCs can be converted into $q$-query LDCs with only a constant loss in rate and preserving other parameters. Below we define LCCs and LDCs formally.
\begin{definition}[Locally Correctable Code]
Let $\Sigma$ be some finite alphabet and let $\cC\subseteq \Sigma^n$ be a set of codewords. $\cC$ is called a $(q,\delta,\epsilon)$-LCC if there exists a randomized algorithm $\cA$ such that following is true:
\begin{enumerate}
\item $\cA$ is given oracle access to some $z\in \Sigma^n$ and an input $i\in [n]$. It outputs a symbol in $\Sigma$ after making at most $q$ queries to $z$.
\item  If $z\in \Sigma^n$ is $\delta$-close to some codeword $c\in \cC$ in Hamming distance, then for every $i\in [n]$, $\Pr[\cA^z(i)=c_i]\ge \frac{1}{2}+\epsilon.$
\end{enumerate}
\end{definition}

It is easy to see that LCCs should have large minimum distance.
\begin{lemma}[Lemma 3.2 in~\cite{BG16}]\label{lem:lcc_distance}
If $\cC\subseteq \Sigma^n$ is a $(q,\delta,\epsilon)$-LCC, then $\cC$ has minimum distance $2\delta$ i.e. every two points in $\cC$ are $2\delta$-far in Hamming distance.
\end{lemma}

\begin{definition}[Locally Decodable Code]
Let $\Sigma$ be some finite alphabet and let $\cC:\BC^k\to \Sigma^n$. $\cC$ is called a $(q,\delta,\epsilon)$-LDC if there exists a randomized algorithm $\cA$ such that following is true:
\begin{enumerate}
\item $\cA$ is given oracle access to some $z\in \Sigma^n$ and an input $i\in [k]$. It outputs a bit after making at most $q$ queries to $z$.
\item  If $z\in \Sigma^n$ is $\delta$-close to a codeword $\cC(x)$ in Hamming distance for some $x\in \BC^k$, then for every $i\in [k]$, $\Pr[\cA^z(i)=x_i]\ge \frac{1}{2}+\epsilon.$
\end{enumerate}
\end{definition}

 We will need the notion of \textit{VC-dimension} for the reduction. 
\begin{definition}
Let $A\subseteq \BC^n$, then the VC-dimension of $A$, denoted by $\vc(A)$ is the cardinality of the largest set $I\subseteq [n]$ which is shattered by $A$ i.e. the restriction of $A$ to $I$, $A|_I=\BC^I$.
\end{definition}
The following lemma due to Dudley(\cite{Dud78}) says that if a set $A\subseteq \BC^n$ has points that are far apart from each other, then it has large VC-dimension.
\begin{lemma}[Theorem 14.12 in \cite{LT13}]\label{lem:vcdim}
Let $A\subseteq \BC^n$ such that for every distinct $x,y\in A$, $\norm{x-y}_{\ell_2}\ge \epsilon \sqrt{n}$. Then $$\vc(A)\ge \Omega\left(\frac{\log|A|}{\log(2/\epsilon)}\right).$$
\end{lemma}
We are now ready to prove the reduction from LCCs to LDCs.

\begin{theorem}
Let $\cC\subseteq \Sigma^n$ be a $(q,\delta,\epsilon)$-LCC, then there exists a $(q,\delta,\epsilon)$-LDC $\cC':\BC^k \to \Sigma^n$ with $$k=\Omega\left(\frac{\log |\cC|}{\log(1/\delta)}\right).$$
\end{theorem}
\begin{proof}
Wlog let us assume $\Sigma=\BC^s$. Let $\cC_0:\BC^s\to \BC^t$ be an error correcting code with distance $\delta_0$ which is some fixed constant. We can extend $\cC_0:\Sigma^n\to \BC^{nt}$ as $$\cC_0(z_1,\cdots,z_n)=(\cC_0(z_1),\cdots,\cC_0(z_n)).$$ By Lemma~\ref{lem:lcc_distance}, every two points in $\cC$ are $2\delta$-far in Hamming distance, it is easy to see that in the concatenated code $\cC_1=\cC_0\circ\cC\subseteq \BC^{tn}$ every two points are $2\delta\cdot \delta_0$ far apart in Hamming distance. So every two points in $\cC_1$ are separated by $\epsilon\sqrt{nt}$ distance in $\ell_2$ norm where $\epsilon=\sqrt{2\delta\delta_0}$. So by Lemma~\ref{lem:vcdim}, $$\vc(\cC_1)\ge \Omega\left(\frac{\log|\cC_1|}{\log(2/\epsilon)}\right)=\Omega\left(\frac{\log |\cC|}{\log(1/\delta)}\right).$$ Therefore there exists a set $I\subseteq [nt]$ of size $k=\vc(\cC_1)$ such that $\cC_1|_I=\BC^I$.

Now define $\cC':\BC^I\to \Sigma^n$ as follows: $\cC'(x)=z$ where $z\in \cC$ is chosen such that $\cC_0(z)|_I=x$ (if there are many such $z$, you can choose one arbitrarily). So the image $\cC'(\BC^I)\subseteq \cC$. Now we claim that $\cC'$ is an $q$-query LDC. Given a word $r\in \Sigma^n$ which is $\delta$-close to $\cC'(x)$, say we want to decode the $i^{th}$ message coordinate $x_i$. Suppose $i$ belongs to the $j^{th}$ block of $(\BC^t)^n$ for some $j\in [n]$. The local decoder of $\cC'$ will run the local corrector of $\cC$ to correct the $j^{th}$ coordinate of $r$ and apply $\cC_0$ to find the required bit $x_i$. So the local decoder for $\cC'$ makes at most $q$ queries and the probability that it outputs $x_i$ correctly is at least $1/2+\epsilon.$
\end{proof}
\section{Decomposition into expanding subgraphs}\label{sec:decomposition}
The goal of this section is to develop a decomposition lemma that approximately partitions any directed graph into a collection of disjoint expanding subgraphs. We use the following notion of edge expansion:

\begin{definition}
A directed graph $G=(V,E)$ is an {\em $\alpha$-edge expander} if for every nonempty $S\subset V$, $$|E(S,V\setminus S)|\ge \alpha |S||V\setminus S|.$$ Here, $E(A,B)$ is the set of edges going from $A$ to $B$.
\end{definition}

We will need the following degree form of Szemer\'edi regularity lemma which can be derived from the usual form of Szemer\'edi regularity lemma for directed graphs proved in~\cite{AS04}.
\begin{definition}
Let $G=(V,E)$ be a directed graph. We denote the indegree of a vertex $v\in V$ by $\indeg_G(v)$ and the outdegree by $\outdeg_G(v)$. Given disjoint subsets $A,B\subset V$, the density $d(A,B)$ between $A,B$ is defined as $$d(A,B)=\frac{E(A,B)}{|A||B|}$$ where $E(A,B)$ is the set of edges going from $A$ to $B$. We say that $(A,B)$ is $\eps$-regular if for every subsets $A'\subset A$ and $B'\subset B$ such that $|A'|\ge \eps |A|$ and $|B'|\ge \eps |B|$, $|d(A',B')-d(A,B)|\le \eps.$
\end{definition}
Note that the order of $A,B$ is important in the definition of an $\eps$-regular pair.

\begin{lemma}[Szemer\'edi regularity lemma for directed graphs (see Lemma 39 in~\cite{Tay14})]\label{lem:szemeredi_regularity_directed}
For every $\epsilon>0$, there exists an $M(\eps)>0$ such that the following is true.
Let $G=(V,E)$ be any directed graph on $|V|=n$ vertices and let $0<d<1$ be any constant. Then there exists a directed subgraph $G'=(V',E')$ of $G$ and an equipartition of $V'$ into $k$ disjoint parts $V_1,\cdots,V_k$ such that
\begin{enumerate}
\item $k\le M(\eps)$.
\item $|V\setminus V'|\le \eps n$.
\item All parts $V_1,\cdots,V_k$ have the same size $m\le \eps n$. 
\item $\outdeg_{G'}(v)\ge \outdeg_G(v)-(d+\eps)n$ for every $v\in V'$.
\item $\indeg_{G'}(v)\ge \indeg_G(v)-(d+\eps)n$ for every $v\in V'$.
\item $G'$ doesn't contain edges inside the parts $V_i$ i.e. $E'(V_i,V_i)=\emptyset$ for every $i$.
\item All pairs $G'(V_i,V_j)$ with $i\ne j$ are $\eps$-regular, each with density $0$ or at least $d$.
\end{enumerate}
\end{lemma}

The regularity lemma above asserts pseudorandomness in the edges going between parts of the partition. For our application and others, it is more natural to require the edges inside each subgraph to display pseudorandomness. As the proof of our Decomposition Lemma shows, we can obtain this from \cref{lem:szemeredi_regularity_directed} with some work.
\begin{lemma}[Decomposition Lemma]\label{lem:expander_decomposition_directed}
Let $G=(V,E)$ be any directed graph on $|V|=n$ vertices. For $0< d< 1$ and $0<\eps<d/6$, there exists a directed subgraph $G'=(V',E')$ and a partition of $V'$ into $U_1,U_2,\cdots,U_K$  where $K\le M(\eps)$ depends only on $\eps$ such that:
\begin{enumerate}
\item $|V\setminus V'|\le 3\eps n$.
\item $\outdeg_{G'}(v)\ge \outdeg_G(v)-(d+3\eps)n$ for every $v\in V'$.
\item $\indeg_{G'}(v)\ge \indeg_G(v)-(d+3\eps)n$ for every $v\in V'$.
\item There are no edges from $U_i$ to $U_j$ where $i>j$.
\item For $1\le i\le K$, the induced subgraph $G'(U_i)$ is either empty or is a $\alpha$-edge expander where $\alpha=\alpha(\eps)>0$.
\end{enumerate}
\end{lemma}

\begin{proof}
We will first apply Lemma~\ref{lem:szemeredi_regularity_directed} to $G$ to get a directed subgraph $G''(V'',E'')$ along with a partition of $V''=V_1\cup \cdots \cup V_k$ as in the lemma where $k\le M(\eps)$. We know that every pair $G''(V_i,V_j)$ is $\eps$-regular with density $0$ or at least $d$. Let us construct a reduced directed graph $R([k],E_R)$ where $(i,j)\in E_R$ iff $G''(V_i,V_j)$ has density at least $d$. Now $R$ has a partition into strongly connected components say given by $[k]=S_1\cup \cdots \cup S_K$ where $K\le M(\eps)$ and $S_1,S_2,\cdots,S_K$ are in topological ordering i.e. there are no edges from $S_i$ to $S_j$ when $i>j$. We will find a large subset $V_j'\subset V_j$ for each of the parts such that $|V_j\setminus V_j'|\le 2\eps |V_j|$ and define $U_i=\cup_{j\in S_i}V_j'$. Our final vertex set will be $V'=\cup_{i=1}^K U_i$ and the graph $G'$ will be the subgraph $G''(V')$. We have $$|V\setminus V'|\le |V\setminus V''| + \sum_{i=1}^k |V_i \setminus V_i'|\le 3\eps n.$$  For every $v\in V'$, $$\indeg_{G'}(v)\ge \indeg_{G''}(v)-\sum_{i=1}^k |V_i \setminus V_i'| \ge \indeg_G(v)- (d +\eps)n-2\eps n=\indeg_G(v)-(d+3\eps)n.$$ Similarly $\outdeg_{G'}(v)\ge \outdeg_G(v)-(d+3\eps)n$. Because the components $S_1,\cdots,S_k$ are in topological ordering with respect to the reduced graph $R$, we cannot have any edges between $U_i$ and $U_j$ where $i>j$.

Now we describe how to find these subsets $V_j'$ where $j\in S_i$ for each of the $S_i$'s and also show the required expansion property.
If $S_i$ is a singleton set i.e. $S_i=\{j\}$ for some $j$, then we just define $V_j'=V_j$. In this case, we will have $U_i=V_j$ and the subgraph $G'(U_i)$ will be empty. If $|S_i|>1$, the subgraph $R(S_i)$ is strongly connected with at least two vertices. So every vertex $j\in S_i$ has at least one outgoing neighbor and one incoming neighbor in $R(S_i)$; choose one outgoing neighbor and call it $N^+(j)$ and choose one incoming neighbor and call it $N^-(j)$. 
Let $V_j'\subset V_j$ be the subset of vertices with at least $(d-\eps)|V_{N^+(j)}|$ outgoing neighbors in $V_{N^+(j)}$ and at least $(d-\eps)|V_{N^-(j)}|$ incoming neighbors in $V_{N^-(j)}$. We will now show that $|V_j\setminus V_j'|\le 2\eps |V_j|$. Let $B_j^+\subset V_j$ be the set of vertices with less than $(d-\eps)|V_{N^+(j)}|$ neighbors in $V_{N^+(j)}$. Define $B_j^-\subset V_j$ similarly. We have $V_j'=V_j\setminus (B_j^+\cup B_j^-)$. So it is enough to show $|B_j^+|\le \eps |V_j|$ and $|B_j^-|\le \eps |V_j|$. 

Consider the $\eps$-regular pair $(V_j,V_{N^+(j)})$ which has density at least $d$. The density between $B_j^+$ and $V_{N(j)}$ can be bounded as $$\frac{|E''(B_j^+,V_{N^+(j)})|}{|B_j^+||V_{N^+(j)}|} < d-\eps\le d(V_j,V_{N^+(j)})-\eps.$$
By $\eps$-regularity of $G''(V_j,V_{N^+(j)})$, we must have $|B_j^+|\le \eps |V_j|$ as required. Similarly we have $|B_j^-|\le \eps |V_j|$.

Now we need to show that $G'(U_i)$ is an $\alpha$-edge expander. Let $A\subset U_i$. For $j\in S_i$, define $A_j=A\cap V_j'$ and $\barA_j=V_j'\setminus A$ and let $\barA=U_i\setminus A$. We want to show that $E'(A,\barA)\ge \alpha|A||\barA|$ for some constant $\alpha(\eps)>0$. We have three cases:\\
\textbf{Case 1:} $\exists j,\ell\in S_i$ such that $|A_j|\ge 2\eps |V_j'|$ and $|\barA_\ell|\ge 2\eps |V_\ell'|$.\\
Label vertices of $R(S_i)$ with $\mathcal{A}$ if $|A_j|\ge 2\eps |V_j'|$ and also with a label $\mathcal{\barA}$ if $|\barA_j|\ge 2\eps |V_j'|$.\footnote{Some vertices can get both labels, but every vertex will get at least one label.} Every vertex should get at least one of the labels and $j$ has label $\mathcal{A}$ and $\ell$ has label $\mathcal{\barA}$. Since $|S_i|>1$, we can assume with out loss of generality that $j\ne \ell$. Since the graph $R(S_i)$ is strongly connected, there is a directed path from $j$ to $\ell$. On this path, there must exist two adjacent vertices $p,q\in S_i$ such that $p$ has label $\mathcal{A}$, $q$ has label $\mathcal{\barA}$ and there is an edge from $p$ to $q$ in $R(S_i)$. We have $$|A_p|\ge 2\eps |V_p'|\ge 2\eps (1-2\eps) |V_p|\ge \eps |V_p|$$ and similarly $|\barA_q|\ge  \eps |V_q|$. By $\eps$-regularity of $G''(V_p,V_q)$, we can lower the bound the number of edges between $A$ and $\bar{A}$ as follows:
$$|E'(A,\barA)|\ge |E''(A_p,\barA_q)| \ge (d-\eps)|A_p||\barA_q| \ge \eps^2 (d-\eps) n^2/k^2\ge \alpha_0 |A||\barA|$$
where $\alpha_0(\eps)=5\eps^3/M(\eps)^2$ is some constant depending on $\eps$.\\
\textbf{Case 2:} For every $j\in S_i$, $|A_j|<2\eps |V_j'|$. \\
By averaging there exists some $j\in S_i$ such that $|A_j|\ge |A|/|S_i|\ge |A|/k$. We know that every vertex in $V_j'$ has at least $(d-\eps)|V_{N^+(j)}|$ out neighbors in $V_{N^+(j)}$, out of these at least $$(d-\eps)|V_{N^+(j)}|-|V_{N^+(j)}\setminus V_{N^+(j)}'|-|A_{N^+(j)}|\ge (d-5\eps)|V_{N^+(j)}|$$ should lie in $\barA_{N^+(j)}$. So we can bound the expansion as follows:
$$|E'(A,\barA)|\ge |E''(A_j,\barA_{N^+(j)})| \ge (d-5\eps)|V_{N^+(j)}||A_j| \ge (d-5\eps)\frac{n}{k}\frac{|A|}{k} \ge\alpha_1 |A||\barA|$$
where $\alpha_1=\eps/M(\eps)^2$ is some constant depending only on $\eps$.\\
\textbf{Case 3:} For every $j\in S_i$, $|\barA_j|<2\eps |V_j'|$. \\
This is very similar to Case 2. By averaging there exists some $j\in S_i$ such that $|\barA_j|\ge |\barA|/|S_i|\ge |\barA|/k$. Every vertex in $V_j'$ has at least $(d-\eps)|V_{N^-(j)}|$ incoming neighbors in $V_{N^-(j)}$, out of these at least $$(d-\eps)|V_{N^-(j)}|-|V_{N^-(j)}\setminus V_{N^-(j)}'|-|\barA_{N^-(j)}|\ge (d-5\eps)|V_{N-(j)}|$$ should lie in $A_{N^-(j)}$. So,
$$|E'(A,\barA)|\ge |E''(A_{N^-(j)},\barA_j)| \ge (d-5\eps)|V_{N^-(j)}||\barA_j| \ge (d-5\eps)\frac{n}{k}\frac{|\barA|}{k} \ge\alpha_1 |A||\barA|$$
where $\alpha_1=\eps/M(\eps)^2$.

Finally we can take $\alpha=\min(\alpha_0,\alpha_1)$, to get the required expansion property.
\end{proof}

The decomposition lemma allows to give an alternative proof for Claim~\ref{claim:graph theoretic}, with worse dependency on $\tau$.
To account for that, we restate Claim~\ref{claim:graph theoretic} and replace $O((1/\tau^4) \cdot \log n)$ with $O_{\tau}(\log n)$.
\begin{claim}
Let $S$ be a set of size  $O_{\tau}(\log n)$ such that $R(S) = T_1$.	
Then, $S$ can be extended by at most  $O_{\tau}(\log n)$ elements, such that $R(S) = V$.	
\end{claim}
\begin{proof}
Let $\{\cM_v:v\in T_2\}$ be the matchings obtained from Lemma~\ref{lem:lcc_matching}, we know that $|\cM_v|\ge \frac{\tau}{4}n$ for each $v\in T_2$. We will construct a directed graph $G(V,E)$ where $V=[n]$ and $E$ is defined as follows. For every $v\in T_2\setminus R(S)$ and every edge $
\{i,j\}\in \cM_v$, add directed edges $(i,v),(j,v)$ to $E$. Thus there is a natural pairing among the directed edges of $G$, we will call $(j,v)$ the \textit{pairing edge} of $(i,v)$ and vice versa. $\{i,j\}$ is called the \textit{matching edge} corresponding to the pair $(i,v),(j,v)$. Since each matching $\cM_v$ has size $\ge \tau n/4$, we have $\indeg_G(v)\ge \delta n$ where $\delta:=\tau/2$ for every $v\in T_2\setminus R(S)=V\setminus R(S)$.

 We now apply Lemma~\ref{lem:expander_decomposition_directed} to get a subgraph $G'=(V',E')$ as described in the lemma where we will choose $\eps=\delta/100$ and $d=\delta/10$. Let $V'=U_1\cup \cdots \cup U_K$ be the partition of $G'$ as described in the lemma where $K\le M(\delta)$. Let $V_0=[n]\setminus V'$ be the remaining vertices, we have $|V_0|\le 3\eps n$. Each vertex $v\in V'\cap (T_2\setminus R(S))$ has $\indeg_{G'}(v)\ge (\delta-d-3\eps)n$. We also know that each sub-graph $G'(U_i)$ is either empty or is an $\alpha$-edge expander for some constant $\alpha(\eps)>0$.

Note that $S$ already has $O_\tau(\log n)$ vertices. We will now grow the set $S$ of coordinates queried by $\cP$ iteratively, adding one at a time. Algorithm~\ref{alg:growing_S} gives the procedure for growing the set $S$.
\begin{algorithm}
\caption{Algorithm for growing $S$}
\label{alg:growing_S}
\begin{algorithmic}
\FOR {$i=1$ \TO $K$}
\STATE  \textbf{Intialization:} Pick one vertex from $U_i$ and add it to $S$. 
\WHILE {$U_i\nsubseteq R(S)$} 
\STATE{Pick any $v\in V\setminus R(S)$ such that adding it to $S$ will add the maximum number of vertices in $U_i\setminus R(S)$ to $R(S)$.}
\ENDWHILE
\ENDFOR
\end{algorithmic}
\end{algorithm}

We will finish the analysis in a series of claims. Let us start with a simple claim about properties of $R(S)$.


\begin{claim}\label{claim:R(S)_properties}
$R(S)$ has the following properties:
\begin{enumerate}
\item If $i,j\in R(S)$ and $(i,j)\in \cM_k$ then $k\in R(S)$.
\item For every edge $(i,k)\in E(R(S),V\setminus R(S))$, there is a unique $j\in V\setminus R(S)$ such that $(i,j)\in \cM_k$.
\end{enumerate}
\end{claim}
\begin{proof}
(1) We can recover $c_i,c_j$ from $c|_S$ and then use them to recover $c_k$ since by Lemma~\ref{lem:lcc_matching}, there exists an algorithm $\cR^k_{i,j}$ such that for every $c\in \cC$, $\cR^k_{i,j}(c_i,c_j)=c_k$.\\
(2) Let $(j,k)$ be the pairing edge of $(i,k)$ so that $(i,j)\in \cM_k$. Now $j$ cannot be in $R(S)$ because of (1).
\end{proof}

Algorithm~\ref{alg:growing_S} should terminate, since $|U_i\cap R(S)|$ increases by at least one in every iteration of the while loop. At the end of the procedure we clearly have $V'=U_1\cup \cdots \cup U_K \subset R(S)$. In fact, we can claim that at the end of the procedure $R(S)=V$ i.e. we can recover all the coordinates of $c$ from $c|_S$.
\begin{claim}
After Algorithm~\ref{alg:growing_S} terminates, $R(S)=V=[n]$.
\end{claim}
\begin{proof}
 After Algorithm~\ref{alg:growing_S} terminates, we have $V'\subset R(S)$. Now we are left with $V_0=V\setminus V'$ where we know that $|V_0|\le 3\eps n$. Now if $w\in V_0\setminus R(S)$ then $w\in T_2\setminus R(S)$ since $T_1\subset R(S)$. Therefore $\indeg_G(w)\ge \delta n$. So there must be $\delta n - |V_0| \ge (\delta -3\eps)n$ incoming edges from $V'$ to $w$. So two of these incoming edges must from a pair and so we have $w\in R(S)$ by part (1) of Claim \ref{claim:R(S)_properties}. Therefore $V_0 \subset R(S)$ as well.
\end{proof}

\begin{claim}
Algorithm~\ref{alg:growing_S} terminates after $O_\delta(\log n)$ rounds.
\end{claim}
\begin{proof}
We just need to show that the while loop runs for $O_\delta(\log n)$ rounds for each $i\in [K]$ since the outer for loop runs for $K$ times where $K\le M(\delta)$. There are two cases:\\
\textbf{Case 1:} The subgraph $G'(U_i)$ is empty.\\
In this case, we will show that $U_i$ must already be contained in $R(S)$. Suppose not, let $w\in U_i\setminus R(S)$, we have $\indeg_{G'}(w)\ge (\delta - d- 3\eps)n$. Moreover, all of these incoming edges come from $U_1,\cdots, U_{i-1}$ (note that this means $i>1$ for this case to happen). Therefore there must be two incoming edges from $U_1\cup \cdots \cup U_{i-1}$ which form a pair i.e. there exists $u,v \in U_1\cup \cdots \cup U_{i-1}$ such that $(u,v)\in \cM_w$. So by part (1) of Claim \ref{claim:R(S)_properties}, $w\in R(S)$. This is a contradiction.\\
\textbf{Case 2:} The subgraph $G'(U_i)$ is an $\alpha$-edge expander.\\
 If $U_i \nsubseteq R(S)$, we will show that after the end of the iteration $t_i:=|R(S)\cap U_i|$ increases by a factor of $(1+\epsilon \alpha)$. This will prove the required claim because $t_i$ is upper bounded by $n$. 

We first claim that $|U_i\setminus R(S)|\ge \eps n$. Suppose this is not true i.e. $|U_i\setminus R(S)|\le \eps n$.  Let $w\in U_i \setminus R(S)$. We know that $w$ has $\indeg_{G'}(w) \ge (\delta - d - 3\eps)n$ incoming edges in $G'$. Since no edges come from $U_j$ for $j>i$, at least $(\delta - d - 3\eps)n - |U_i\setminus R(S)| \ge (\delta -d -4\eps)n$ of them come from $U_1 \cup \cdots \cup  U_{i-1} \cup (U_i \cap R(S)) \subset R(S)$. Therefore two of the incoming edges must form a pair and so $w\in R(S)$ which is a contradiction.

Since $G'(U_i)$ is an $\alpha$-edge expander, we have $$E(U_i\cap R(S), U_i \setminus R(S))\ge \alpha t_i |U_i \setminus R(S)| \ge \alpha \epsilon t_i n.$$ By part (2) of Claim~\ref{claim:R(S)_properties}, each edge from $U_i\cap R(S)$ to $U_i \setminus R(S)$ corresponds to a matching edge between $U_i \cap R(S)$ and $V\setminus R(S)$ and it belongs to a matching which corresponds to a vertex in $U_i \setminus R(S)$. Therefore there are at least $\alpha \epsilon t_i n$ matching edges between $U_i \cap R(S)$ and $V\setminus R(S)$ which belong to $\cup_{w\in U_i\setminus R(S)}\cM_w$; by averaging there exists $v\in V\setminus R(S)$ which is incident to $\alpha \epsilon t_i n /|V\setminus R(S)|\ge \alpha \epsilon t_i$ of these matching edges. So adding this $v$ to $S$ will add $\alpha \epsilon t_i$ new vertices of $U_i\setminus R(S)$ to $R(S)$, increasing $t_i$ by a factor of $(1+\alpha \epsilon)$.
\end{proof}
\end{proof}

\end{document}